\documentclass[11pt]{article}

\usepackage[letterpaper]{geometry}

\usepackage{geometry}
\usepackage{amssymb}
\usepackage{mathrsfs}
\usepackage{amsmath}
\newtheorem{theorem}{Theorem}[section]
\newtheorem{corollary}[theorem]{Corollary}
\newtheorem{lemma}[theorem]{Lemma}
\newtheorem{proposition}[theorem]{Proposition}
\newtheorem{claim}[theorem]{Claim}
\newtheorem{definition}[theorem]{Definition}

\newtheorem{example}[theorem]{Example}

\def\squarebox#1{\hbox to #1{\hfill\vbox to #1{\vfill}}}
\newcommand{\qed}{\hspace*{\fill}
\vbox{\hrule\hbox{\vrule\squarebox{.667em}\vrule}\hrule}\smallskip}

\newenvironment{proof}{\noindent{\bf Proof:~~}}{\(\qed\)}

\begin{document}

%\toappear{}
\title{Breaking the Logarithmic Barrier for Truthful \\Combinatorial Auctions with Submodular Bidders}

\author{
Shahar Dobzinski\thanks{Weizmann Insittute of Science. Incumbent of the Lilian and George Lyttle Career Development Chair. Supported in part by the I-CORE program of the planning and budgeting committee and the Israel Science Foundation 4/11 and by EU CIG grant 618128.}
}
%\author{Shahar Dobzinski\footnote{Weizmann Institute of Science. Incumbent of the Lilian and George Lyttle Career Development Chair. Supported in part by the I-CORE program of the planning and budgeting committee and the Israel Science Foundation 4/11 and by EU CIG grant 618128.}}

\maketitle
\begin{abstract}
We study a central problem in Algorithmic Mechanism Design: constructing truthful mechanisms for welfare maximization in combinatorial auctions with submodular bidders. Dobzinski, Nisan, and Schapira provided the first mechanism that guarantees a non-trivial approximation ratio of $O(\log^2 m)$ [STOC'06], where $m$ is the number of items. This approximation ratio was subsequently improved to $O(\log m\log \log m)$ [Dobzinski, APPROX'07] and then to $O(\log m)$ [Krysta and Vocking, ICALP'12].

In this paper we develop the first mechanism that breaks the logarithmic barrier. Specifically, the mechanism provides an approximation ratio of $O(\sqrt {\log m})$. Similarly to previous constructions, our mechanism uses polynomially many value and demand queries, and in fact provides the same approximation ratio for the larger class of XOS (a.k.a. fractionally subadditive) valuations.

We also develop a computationally efficient implementation of the mechanism for combinatorial auctions with budget additive bidders. Although in general computing a demand query is NP-hard for budget additive valuations, we observe that the specific form of demand queries that our mechanism uses can be efficiently computed when bidders are budget additive.
\end{abstract}

%\thispagestyle{empty}

%\newpage
%\setcounter{page}{1}
%\newpage
\section{Introduction}

The economic field of \emph{Mechanism Design} mainly deals with games where each strategic participant privately holds some information. The paradigmatic example is a single item auction, where bidders are interested in an item that is for sale. The value of each bidder for the item is unknown to the other bidders. The mechanism design question is to find an auction format that will achieve a certain social goal. Archetypal examples are Vickrey's second-price auction \cite{Vic61} that maximizes the welfare and Myerson's revenue maximizing auctions \cite{mye81}.

Since the introduction of these classic constructions, we have witnessed the emergence of complex markets such as eBay and Amazon with millions of items that are for sale. Furthermore, auctions have become significantly larger and complicated than before. Examples include spectrum auctions with revenue measured in billions of dollars \cite{C02} as well as more recent ones such as the FCC incentive auctions \cite{MS14}. These markets introduce new challenges that can be very coarsely classified into two. The first type is traditional game theoretic challenges, e.g., bidders have complicated preferences over multiple bundles of items and their private information can no longer be represented by a single number, as in the single item auction case. This considerably limits the set of tools available to the designer. The second type of challenges is computational considerations: some classic designs for these complex settings may have good game theoretic properties but require solving computationally intractable problems, leading to unacceptable running time.

In a sense, the rise of \emph{Algorithmic} Mechanism Design can be related to the need to simultaneously solve these challenges. A \emph{combinatorial auction} is a quintessential setting in this field. The basic definition involves a set of $M$ heterogeneous items ($|M|=m$) and $n$ bidders. Each bidder $i$ has a valuation function $v_i:2^M\rightarrow \mathbb R$. It is assumed that each valuation $v_i$ is normalized ($v_i(\emptyset)=0$) and non decreasing. The usual goal is to find an allocation of the items $(A_1,\ldots, A_n)$ that maximizes the social welfare\footnote{While welfare maximization is probably the most popular goal in the setting of combinatorial auctions, other goals, e.g., revenue maximization, were also studied.} $\Sigma_iv_i(A_i)$. Since we are interested in algorithms that run in time $poly(n,m)$ and the size of the valuation functions is exponential in $m$, it is common to assume that the valuations are given to us as black boxes that can handle specific types of queries. The two standard queries are value queries (given a bundle $S$, return $v(S)$) and demand queries (given prices $p_1,\ldots, p_m$ return $\arg\max_S v(S)-\Sigma_{j\in S}p_j$). 

To handle the strategic behavior of the bidders we charge each bidder $i$ some payment $p_i$ for the bundle $A_i$ he received. We are looking for \emph{truthful} mechanisms, where the profit $v_i(A_i)-p_i$ of each bidder is maximized when answering queries according to his true valuation.

Numerous variations on the basic problem were studied (see, e.g., the survey \cite{BN07} and references within), but most of them demonstrate the basic clash that is in the heart of Algorithmic Mechanism Design: the VCG mechanism is a truthful mechanism that maximizes the welfare but requires finding the welfare maximizing allocation, which is usually NP-hard. On the other hand, good constant factor approximation algorithms exist but are not truthful. The goal is therefore to design truthful mechanisms with approximation ratios close to what is possible from a pure algorithmic point of view that completely ignores incentives issues.

\subsection{The Main Result: Combinatorial Auctions with Submodular Valuations}

The case where all valuations are submodular (for every $S$ and $T$, $v(S)+v(T)\geq v(S\cup T)+v(S\cap T)$) stands out as a showcase for the power and limitations of computationally efficient truthful mechanism design. The pure algorithmic aspect of the problem has received much attention as well (e.g., \cite{FV06, LLN01, MSV08, KLMM05}). Of particular importance is Vondrak's continuous greedy algorithm \cite{V08} that was initially developed for this setting and its numerous extensions to other problems and follow-ups. The algorithmic situation is quite well understood and can be summarized as follows: there is an $\frac {e} {e-1}$-approximation algorithm that uses polynomially many value queries, and this ratio is tight \cite{KLMM05,MSV08}. If the more stronger demand queries are allowed then it is possible to break the $\frac e {e-1}$-barrier and achieve an approximation ratio of $\frac {e} {e-1}-10^{-6}$ \cite{FV06}, but it is impossible to get an approximation ratio better than $\frac {2e} {2e-1}$ with polynomially many queries \cite{DV12}. 

Much less is known about the approximation ratio achievable by polynomial time truthful mechanisms. If access to the valuations is restricted to value queries, then deterministic mechanisms cannot achieve an approximation ratio of $m^{\frac 1 2 -\epsilon}$ \cite{D11}, which matches the ratio obtained by \cite{DNS05}. It is therefore natural to consider randomized mechanisms. There are two important notions of randomized truthful mechanisms. \emph{Truthful in expectation} mechanisms guarantee that bidding truthfully maximizes the \emph{expected} profit. In particular, these mechanisms are inapplicable when bidders are not risk neutral. In contrast, \emph{universally truthful} mechanisms are simply a probability distribution over deterministic mechanisms, and thus it is always better to bid truthfully, regardless of the attitude towards risk.

The first truthful mechanism with non-trivial guarantee was obtained by \cite{DNS06}, which shows that there is a randomized universally truthful mechanism that guarantees an approximation ratio of $O(\log ^2 m)$. This approximation ratio was later improved to $O(\log m\log \log m)$ \cite{D07} and then to $O(\log m)$ \cite{KV12}. All the above mentioned mechanisms use both demand and value queries.

There were reasons to believe that the ``correct'' approximation ratio for this problem is $O(\log m)$ and that the lack of improvement is due to our usual inability to prove impossibility results on the power of truthful polynomial time mechanisms. For example, both \cite{DNS06} and \cite{D07} essentially use a fixed price auction with the same price for each item, and one can show that such auctions cannot guarantee more than a logarithmic approximation. Furthermore, the analysis is based on a comparison to a ``revenue benchmark'', and it is known that the gap between the welfare and the revenue might be logarithmic (the so called equal revenue distribution). In a similar vein, the analysis of the algorithm of \cite{KV12} is essentially based on the observation that one can approximate the welfare well in ``easy'' instances when there are logarithmically many copies of each good. Unfortunately, the welfare gap between an easy instance and the same instance with only one copy of each good can be logarithmic.

Even considering the weaker concept of truthfulness in expectation does not seem to help much. The notable positive result here restricts the valuations to be sum of matroid rank functions. In this case, there is a truthful in expectation that provides an approximation ratio of $\frac e {e-1}$ \cite{DRY11}. However, for general submodular valuations the best known bound is only $O(\frac {\log m} {\log \log m})$ \cite{DFK10} and even this ratio is obtained by a mechanism that uses non-standard queries.

Despite this, in this paper we are able to break the logarithmic barrier:

\vspace{0.1in}\noindent \textbf{Theorem:} There exists a randomized universally truthful algorithm for combinatorial auctions with XOS valuations that achieves an expected approximation ratio of $O(\sqrt {\log m})$ to the social welfare. The algorithm makes $poly(m,n)$ value and demand queries\footnote{The probability of success for which the mechanism provides a good approximation ratio was also studied in few cases \cite{DNS06,D07}. In our case, using standard arguments we get that the probability that the mechanism will provide at least half of its expected value is $\Omega(\frac 1 {\sqrt {\log m}})$. Notice that unlike traditional algorithm design running a truthful mechanism more than once usually destroys its incentive properties. In principle, we believe that it should be possible to imrpove the probability of success of our mechanisms using techniques similar to those of \cite{DNS06,D07}. However, we do not push this direction in the current manuscript in order to keep our construction a bit simpler. Finally, we note that the probability of success of the mechanism obtained in \cite{KV12} is $\Theta(\frac 1 {\log m})$.}.

\vspace{0.1in}\noindent Notice that the new mechanism (as well as all previous ones) actually works for the larger class of XOS (a.k.a. fractionally subadditive) valuations. A valuation $v$ is \emph{additive} if for every bundle $S$ we have that $v(S)=\Sigma_{j\in S}v(\{j\})$. A valuation $v$ is \emph{XOS} if there exist additive valuations $a_1,\ldots, a_t$ such that for every bundle $S$, $v(S)=\max_r a_r(S)$. Each $a_r$ is a \emph{clause} of $v$. If $a\in\arg\max_r a_r(S)$ then $a$ is a \emph{maximizing clause} of $S$ and $a(j)$ is the \emph{supporting price} of item $j$ in this maximizing clause.

\subsection{Intuition for the Mechanism}

A basic construction used in our mechanism is a \emph{fixed price auction}. In this auction there is a price $p_j$ for each item $j$. Initially, $M'=M$. Bidders arrive one by one, in an arbitrary order. Each bidder $i$ that arrives to the auction takes some bundle $S_i\subseteq M'$ that maximizes his profit: $S_i\in \arg\max_S[ v_i(S)-\Sigma_{j\in S}p_j]$. Let $M'=M'-S_i$ and consider the next bidder in the order. At the end of the auction each bidder $i$ receives the set $S_i$ and pays $\Sigma_{j\in S_i}p_j$. Observe that the fixed price auction is truthful as long as no participating bidder affects $p_1,\ldots,p_m$. Also note that to implement the auction we need one demand query for each participating bidder.

Let us first see how to obtain a logarithmic approximation to the social welfare using a fixed-price auction\footnote{This slightly improves over the results of \cite{DNS06,D07} in which an $O(\log m\log\log m)$ approximation ratio was obtained by using a fixed price auction. The logarithmic approximation of \cite{KV12} is obtained via a more complicated auction.}, and then discuss how to improve the approximation ratio. For simplicity of presentation assume that all values in the additive valuations that the XOS valuations are constructed from are integers in $\{1,2,4,8, \ldots, m\}$.

Let $(O_1,\ldots, O_n)$ be an optimal allocation and $a_i$ be the maximizing clause of $O_i$ in the valuation $v_i$ of bidder $i$. For each bidder $i$ and $j\in O_i$, let $q_j=a_i(\{j\})$. Notice that $OPT=\Sigma_iv_i(O_i)=\Sigma_jq_j$. The basic idea is to find prices $p_1,\ldots, p_m$ such that a fixed price auction will output an allocation with high welfare. Towards this end, we say that we have ``correctly guessed'' the price of item $j$ if $q_j=2\cdot p_j$. It was already observed (e.g., \cite{DNS06, D07} -- see also Lemma \ref{lemma:xos:revenue} of this paper) that the welfare obtained by a fixed price auction with prices $p_1,\ldots,p_m$ is at least $2\cdot\Sigma_{j|q_j=2\cdot p_j}q_j$. Therefore, our approach is to maximize the total value of correctly guessed items. Notice that in the analysis we conservatively ignore the contribution of items that were not correctly guessed regardless of whether their price is ``too high'' or ``too low''. 

We first explore the possibility of using a uniform price for all items. Partition the optimal allocation intos: put all items $j$ with $q_j=r$ in bin $G_r$. For convenience, we will say that the price associated with bin $G_r$ is $r$. Let $cont(G_r)=|\{j|q_j=r\}|\cdot r$ denote the contribution of bin $G_r$ to the optimal solution. As discussed above, the welfare of a fixed price auction with price per item $r/2$ is $\Omega(cont(G_r))$. Thus, since $\Sigma_rcont(G_r)=OPT$ and since there are $\log m$ bins by assumption, the expected approximation ratio of a fixed-price auction with a random price $r/2$ is $O(\log m)$.

To improve the approximation ratio, we group every $\sqrt {\log m}$ consecutive bins to a single chunk, so that we have $\sqrt {\log m}$ disjoint chunks. For each chunk $C_k$, denote by $val(C_k)=\Sigma_{G_r\in C_k}cont(G_r)$ the sum of contributions of bins that are in $C_k$. Notice that we cannot guarantee now that there is a fixed price auction with a uniform price that generates welfare of $val(C_k)$, since a chunk consists of multiple bins that contain items with different $q_j$'s. However, we might get lucky: let $r_k$ be the smallest price associated with a bin in $C_k$. If a fixed price auction with uniform price $r_k/2$ returns welfare of $\Omega(val(C_k))$, we say that $C_k$ is easily approximable. Let $NE$ denote the set of chunks that are not easily approximable.

In the very lucky case, $\Sigma_{C_k\notin NE}val(C_k)\geq \frac {OPT} 2$. Similarly to before, in this case we get an approximation ratio of $O(\sqrt {\log m})$ by choosing uniformly a random a chunk $C_k$ out of the $\sqrt {\log m}$ chunks and running a fixed price auction with price $r_k/2$ for all items.

The situation is more complicated when $\Sigma_{C_k\in NE}val(C_k)> \frac {OPT} 2$. Our goal now is to run a fixed price auction that uses multiple prices. Key in this plan is Lemma \ref{lemma:random-set-posted-price} that considers a fixed price auction with a uniform price $r_k/2$ \emph{where bidders arrive in a random order}. Let the output of this auction be $(T_1,\ldots, T_n)$, so $\mathcal T=\cup_iT_i$ is the set of items that were allocated in the fixed-price auction. The lemma roughly says that if chunk $C_k$ is not easily approximable then $E[\Sigma_{j\in \mathcal T}q_j]=\Omega(val(C_k))$. Notice that this does not imply that the welfare $\Sigma_iv_i(T_i)$ of the fixed price auction is large: in an extreme case all items $j$ in the chunk have $q_j=r_k\cdot 2^{\sqrt {\log m}}$ whereas the fixed price auction allocates all items to some bidder $i$ with $v_i(T_i)=|T_i|\cdot r_k/2$. However, the lemma does imply that the set $\mathcal T$ contains items that contribute much of the value of chunk $C_k$. Therefore, we can ``guess'' the price of all items in $\mathcal T$, hoping to correctly guess the price of items that are also in chunk $C_k$. The point is that there are only $\sqrt {\log m}$ bins in a chunk, so in expectation we correctly guess the price of items that contribute value $\Omega(\frac {val(C_k)} {\sqrt {\log m}})$. 

The overall plan is to run a fixed price auction with price $r_k/2$ on every chunk $C_k$ in an increasing order of $k$, and guess the price of all items that were allocated in the auction. In general, some items that were correctly guessed in some fixed price auction with price $r_k/2$ might be (incorrectly) re-priced in one of the following fixed-price auction with price $r_{k'}/2$ for $k'>k$. However, we are able to bound the loss of such re-pricing by making sure that $r_{k'}>>r_k$. We show, very roughly speaking, that the total number of items allocated in each fixed-price auction shrinks dramatically, and therefore the number of items that are re-priced is bounded.

All this hints that we can correctly guess the price of many items and have that $\Sigma_{j|q_j=2\cdot p_j}q_j=\Omega(\frac {OPT} {\sqrt {\log m}})$. Ideally, we would like to run a fixed price auction with these prices, but recall that a fixed price auction is truthful only when the prices do not depend on the participants` valuations, which is not the case here. Furthermore, how can we determine whether the easily approximable chunks are significant or not? The second issue turns out to be easy to solve, as we can ``guess'' whether the easily approximable chunks are significant or not simply by flipping a random coin. The solution of the first issue is standard, although the analysis is a bit more delicate than the usual: we use only half of the bidders in the fixed price auctions that are used to guess prices. After guessing the $p_j$'s, we run a final fixed-price auction where the participants are the other half of bidders that were not involved in determining the $p_j$'s. This guarantees the truthfulness of the mechanism. The reader is referred to the technical parts for a complete description and analysis of the mechanism.

\subsection{Budget Additive Bidders}

As noted above, our mechanism assumes that the valuations are given as black boxes that can only be accessed via demand queries. However, sometimes the valuations are explicitly given but simulating a demand query might be NP hard. One extensively studied case (e.g., \cite{AM04, CG08, BDFK10}) is when all valuations are budget additive: there exists some $b$ such that for every bundle $S$, $v(S)=\min(b,\Sigma_{j\in S}v(\{j\}))$. A simple reduction from, say, the knapsack problem shows that simulating a demand query is NP hard. Yet, we observe that our mechanism uses demand queries of a very specific form, and these can be computed in polynomial time. Hence:

\vspace{0.1in}\noindent \textbf{Theorem:} There exists a polynomial time randomized universally truthful algorithm for combinatorial auctions with budget additive valuations that achieves an expected approximation ratio of $O(\sqrt {\log m})$.

\subsection{Open Questions}

The obvious question that we leave open in this paper is determining whether $\Theta(\sqrt{ \log m})$ is the best possible approximation ratio in our setting. Our mechanism uses both randomization and demand queries, but -- as far as impossibility results are concerned -- all we know is that deterministic mechanisms that use only value queries cannot obtain an approximation ratio of $O(m^{\frac 1 2 -\epsilon})$ \cite{D11}. This result was extended to randomized mechanisms that use value queries \cite{DV11} (with a constant in the exponent that is smaller than $\frac 1 2$). In particular, we note that nothing is known even about the power of deterministic mechanisms that use demand queries (but we conjecture that there is not much to gain here). The interested reader is referred to \cite{D16b} for possible approaches for proving impossibility results for computationally efficient truthful mechanisms in general, and for mechanisms that use only demand and value queries -- as the mechanism introduced in this paper -- in particular.

We do not know whether our results can be extended to subadditive valuations. The best known $O(\log m\log\log m)$-approximation algorithm \cite{D07} relies on the fact that for every subadditive valuation there is an XOS valuation that $O(\log m)$-approximates it\footnote{A valuation $v$ $\alpha$-approximates a valuation $v'$ if for every bundle $S$  it holds that $v(S)\leq v'(S)\leq \alpha\cdot v(S)$.}. Thus a logarithmic factor loss seems inevitable with this approach. Breaking the logarithmic barrier also for subadditive valuations looks challenging.

\section{Preliminaries}

%\subsection{Combinatorial Auctions }
%
%In a combinatorial auction we have a set $N$ of players ($|N|=n)$ and a set $M$ of different items ($|M|=m$). Each player $i$ has a valuation function $v_i:2^M\rightarrow \mathbb R$. Each $v_i$ is assumed to be normalized ($v_i(\emptyset)=0)$ and non decreasing. The goal is to maximize the social welfare, that is, to find an allocation of the items to players $(A_1,\ldots, A_n)$ that maximizes the welfare: $\Sigma_iv_i(A_i)$. 

\subsection{Allocations and Bins}\label{subsec:allocations}

Let $A=(A_1,\ldots, A_n)$ be some allocation. Let $|A|$ denote $\Sigma_iv_i(A_i)$ rounded up to the nearest power of $2$. Let $\mathcal P=\{\frac {2^k\cdot |A|} { m^2}\}_{k\in \mathbb Z}$. Suppose that all valuations are XOS and let $q_j$ denote the supporting price of $j$ in $A_i$, for every $i$ and item $j\in A_i$. Let $q'_j$ be the maximal value in $\mathcal P$ such that $q_j\geq q'_j$. 

It will be convenient not to work directly with the set $\mathcal P$ but rather with \emph{bins} that reference $|A|$ only indirectly. Specifically, bin $k$ of the allocation $A$ is $k$ is $p(k)=\frac {2^k|A|} { m^2}$ (we will also refer to $p(k)$ as the price associated with bin $k$ in $A$). We say that item $j$ is in bin $k$ if $q'_j=p(k)$. Let $n_k$ be the number of items in bin $k$. The \emph{contribution} of bin $k$ in $A$ is $p(k)\cdot n_k$. %We call a bin $k>0$ \emph{significant}. The next claim justifies this name:

We say that the allocation $T'=(T'_1,\ldots, T'_n)$ is the allocation $T=(T_1,\ldots, T_n)$ restricted to a set of bins $B$ if for each $i$ we have that $T'_i=T_i\cap M_b$, where $M_b$ is the set of items that are in some bin in $B$ in the allocation $(T_1,\ldots, T_n)$. We say that an allocation $T=(T_1,...,T_n)$ is supported by prices $p'_1,\ldots,p'_m$ if for each item $j\in T_i$ we have that $q_j\geq p'_j$, where $q_j$ is the supporting price of item $j$ in $T_i$ according to $v_i$. Note that we do not assume that every item $j$ is allocated in $T$. %If every for some $p$ we have that for all $j$, $p_j=p$ we say that $T$ is supported by $p$.

\subsection{Truthfulness}

Let $V$ be some set of valuations. An $n$-bidder mechanism for combinatorial auctions is a pair $(f,p)$ where $f:V^n \rightarrow \mathcal A$, where $\mathcal A$ is the set of all allocations, and $p=(p_1,\ldots, p_n)$, where $p_i:V^n\rightarrow \mathbb R$. %$f$ might be either randomized or deterministic.

\begin{definition}
Let $(f,p)$ be a deterministic mechanism. $(f,p)$ is \emph{truthful} if for all $i$, all $v_i, v'_i\in V$ and all $v_{-i}\in V^{n-1}$ we have that $v_i(f(v_i,v_{-i})_i)-p_i(v_i,v_{-i})\geq v'_i(f(v'_i,v_{-i})_i)-p(v'_i,v_{-i})$. $(f,p)$ is \emph{universally truthful} if it is a probability distribution over truthful deterministic mechanisms.
\end{definition}

\section{The Mechanism}

We now provide a description of the mechanism. Let $\alpha=\sqrt {\log m}$.

%\subsection*{The Mechanism}

\begin{enumerate}
\item\label{step:second-price} With probability $1/2$ sell the grand bundle $M$ via a second price auction with the participation of all bidders. The mechanism ends in this case with the allocation and prices of this second-price auction.

\item\label{step:sample} Each bidder is assigned independently with equal probability to exactly one of the following three groups: STAT, UNIFORM, and FINAL.

\item\label{step:approximate} Run the greedy algorithm of \cite{LLN01}\footnote{Any other $O(1)$ approximation algorithm that uses a polynomial number of value and demand queries will yield similar results.} with the participation of bidders in STAT only. 

Let\footnote{Throughout this paper, we denote for notational simplicity allocations to a subset of the bidders such as $APX$ as allocations for all $n$ bidders by letting $APX_i=\emptyset$ for each bidder $i\notin STAT$.} $APX=(APX_1,\ldots, APX_n)$ be the output allocation. % and $APX_{STAT}=\Sigma_iv_i(APX_i)$.

\item Partition bins $1, \ldots ,4\log m$ of the allocation APX into $ { \alpha}$ disjoint chunks so that chunk $C_k$ contains the following bins: ${\{(k-1)\cdot \frac {4\log m} {\alpha }+1,\ldots, k\cdot \frac {4\log m} {\alpha }\}}$.

\item\label{step:initalize} Select uniformly at random an integer $r$ from the set $\{1,2,3, \ldots, \frac {4\log m} {\alpha }\}$.

\item\label{step:order} Choose an order $\pi$ over the bidders in UNIFORM uniformly at random. In addition, for every item $j$, let $p_j=0$.

\item\label{step:loop} Consider each chunk $C_k$, in ascending order:

\begin{enumerate}	
	\item Let $p'$ be the smallest price associated with any bin in $C_k$. Run a fixed price auction restricted to bidders in UNIFORM with price $p'/2$ for every item, where the order of the bidders in UNIFORM is $\pi$ (Step \ref{step:order}). Denote the allocation that this fixed-price auction outputs by $T^k=(T^k_1,\ldots, T^k_n)$ and by $p^{i,k}$ the payment of bidder $i$.
	
	\item\label{step:count-and-posted-price} With probability $\frac 1 \alpha$ the mechanism ends with the allocation $T^k$. Each bidder $i$ pays $p^{i,k}$.
	
	\item\label{step:pricing} Let $p''$ be the price of the bin with the $r$'th smallest value in $C_k$. Update the price $p_j$ of every item $j\in \cup_i T^k_i$ to $p_j=p''/2$. 
\end{enumerate} 

\item\label{step:final-posted-price-auction} Run a fixed price auction with prices $p_1,\ldots, p_m$ with the participation of bidders in FINAL only. Output the allocation and prices of this final auction.
\end{enumerate}

Before proceeding to a formal analysis of the mechanism, let us make some brief comments. The heart of the mechanism is clearly the fixed price auctions of Steps \ref{step:loop} and \ref{step:final-posted-price-auction}. The main issue is how to determine the prices that will be used in these auctions in a truthful manner. The standard way to do so is by excluding a random set of bidders from winning items. The valuations of bidders in the excluded set can be queried in order to set the prices in the fixed price auctions. Since the excluded set was chosen at random, the hope is that the excluded set is indeed a ``representative'' sample and thus the answers to the queries are indeed useful. Note that we expect the bidders in the excluded set to answer the queries truthfully, since they will not win any items in any case. In our mechanism, depending on the outcome of the random coins, only bidders in $UNIFORM$ or bidders in $FINAL$ might win some items. The excluded set is bidders in STAT in the first case and bidders in $UNIFORM\cup STAT$ in the second one.

One obvious issue with this approach is that sometimes it is impossible to gather useful information by random sampling. For example, perhaps there is one bidder with a valuation that is very large comparing to the valuations of the other bidders. In this case we have little hope that the fixed price auctions will provide high welfare: even if the significant bidder is participating in the fixed price auction, the excluded set is not a representative sample so the prices used in these auctions might be too low to guarantee a good approximation. The standard way to avoid this is to ``guess'' -- by flipping a random coin -- whether there is a significant contributer to the optimal welfare. If there is, then a good approximation ratio can be guaranteed by running a second price auction on the grand bundle (Step \ref{step:second-price}). 

Bidders in $UNIFORM$ are participating in a sequence of fixed-price auctions. Based on the random coins, they are either allocated the outcome of one random auction, or cannot win items at all (Step \ref{step:loop}). In the latter case, the outcomes of the fixed priced auctions are used to set the prices in the auction with the participation of bidders in $FINAL$ only (Step \ref{step:final-posted-price-auction}). 

Bidders in $STAT$ are only used to compute the allocation $APX$. This allocation is used to determine the prices of the fixed price auctions of Step \ref{step:loop}. This is a good time to note that the only information about the allocation $APX$ that the algorithm uses is its value. This value is used only to determine the range of prices that will be used in the fixed priced auctions. The algorithm does not make use of the specific allocation of items in $APX$ or other properties of it. Furthermore, suppose that we are given that for every item $i$ and bundle $S$ we have that $v_i(S)\in \{0\}\cup [1,\ldots, m]$ (or, almost equivalently, a good estimate of the optimal solution). In this case we could just let the mechanism consider prices $1,2,4,\ldots,  m$ without making any use of bidders in $STAT$.

We now turn to the formal analysis of the mechanism.

\begin{theorem}
The mechanism is universally truthful and can be implemented with polynomially many demand and value queries. It provides an approximation ratio of $O(\sqrt{\log m})$.
\end{theorem}

We first show that the mechanism is truthful and uses a polynomial number of demand queries. For truthfulness, observe that the only way for a bidder to receive some item is by participating in a fixed price auction or a second-price auction. Which of these two auctions is conducted and which bidders are allowed to participate is determined solely be flipping random coins. 

The mechanism is obviously truthful when a second price auction on the grand bundle is conducted. In addition, notice that the price in the fixed-price auctions does not depend on the valuations of the participating bidders. Moreover, despite participating in multiple auctions, bidders in UNIFORM might only win items in one auction that is determined by the outcome of the random coins. 

This already guarantees that truth telling is an ex-post Nash equilibrium. However, truth-telling might not be a dominant strategy: for example, the first bidder $i$ in the order $\pi$ might ``threat'' the second bidder $i'$ in $\pi$ that if bidder $i'$ reports that his demand consists of, say, exactly two items in the first fixed-price auction, bidder $i$ will report that his demand consists of all items in all subsequent auctions (even if his demand is different) thus leaving bidder $i'$ with no items. Such issues can be avoided by a careful implementation, for example, by hiding the answers to the queries of bidders from their predecessors in the order. Alternatively, the mechanism can ask the first bidder in $\pi$ to report his demand for \emph{all} fixed-price auctions simultaneously. Then, based on the answers of the first bidder ask the second bidder in $\pi$ to report his demand for all fixed-price auctions simultaneously, and so on.

To see that the bound on the number of demand queries, observe that we run at most $\alpha+1$ fixed price auctions, and that each such auction requires at most $n$ demand queries. The next two sections are devoted to proving the approximation ratio.

\section{Analysis of the Approximation Ratio: Auxiliary Lemmas}\label{sec-approx-ratio}

Variants of the next two lemmas have already appeared in the literature multiple times, e.g., \cite{DNS06,D07} (proofs in the appendix). In particular, the first lemma is almost an immediate application of the Hoeffding bounds.

\begin{lemma}\label{lemma-statistics}
Let $S$ be a random set of bidders where bidder $i$ is in $S$ with probability $p$ independently of the other bidders. Let $(T_1,\ldots,T_n)$ be some allocation and suppose that for every bidder $i$ we have that $v_i(T_i)\leq \frac {\Sigma_kv_k(T_i)} {R}$. Then, with probability at least $1-2e^{-\frac {p\cdot R} {2}}$:

$$
\Sigma_{i\in S}v_i(T_i) \geq \frac {p\cdot \Sigma_iv_i(T_i)} 2
$$
\end{lemma}

\begin{lemma}\label{lemma:xos:revenue}
Let $T=(T_1,...,T_n)$ be an allocation that is supported by prices $p'_1,\ldots,p'_m$. A fixed price auction with prices $p_j=\frac {p'_j} 2$ generates an allocation $(S_1,\ldots, S_n)$, with $\Sigma_iv_i(S_i)\geq \frac {\Sigma_{j\in \cup_iT_i}p'_j} 2$.
\end{lemma}
The next lemma considers a fixed price auction with a random subset of the bidders that arrive in random order. It compares the quality of the solution to some arbitrary allocation $A$ and shows that we expect to either get an allocation with a welfare close to that of $A$ or that many of the items of $A$ were allocated.

\begin{lemma}\label{lemma:random-set-posted-price}
Let $A=(A_1,\ldots, A_n)$ be an allocation that is supported by $p'_1,\ldots, p'_m$. For every item $j\in A_i$, let $q_j$ be its supporting price. Let $o=\Sigma_i\Sigma_{j\in A_{i}}q_j$.
Let $N'$ be a random set of bidders where each bidder is in $N'$ independently at random with probability $r$. Let $T=(T_1, \ldots, T_n)$ be the random variable that denotes the allocation of the fixed price auction with price $p_j=\frac {p'_j} 2$ for every item $j$ when $N'$ is constructed at random as above and the order over bidders in $N'$ is chosen uniformly at random. Let $c=1-\frac {\Sigma_i\Sigma_{j\in A_i\cap(\cup_kT_k)}q_j} {o}$. Then $E[\Sigma_iv_i(T_i)] \geq \frac {o\cdot E[c]\cdot r} {4}$, where expectations are taken over the random choices of $N'$ and its internal order.
\end{lemma}
\begin{proof}
Fix $N'$ and its internal order. Consider a fixed-price auction as above and denote by $(T_1,\ldots ,T_n)$ the allocation of the fixed-price auction. Let $W=\cup_{i\in N'}A_i-\cup_{i\in N'}T_i$ be the set of items that are allocated to bidders in $N'$ in $(A_1,\ldots, A_n)$ but were not allocated in the fixed price auction. Our first step is to bound the social welfare of the fixed price auction using $W$:

\begin{claim}\label{claim-contribution-of-W}
$\Sigma_iv_i(T_i)\geq \Sigma_{j\in W}q_j/2$.
\end{claim}
\begin{proof}
Consider some bidder $i\in N'$ arriving to the auction. Let $A'_i=W\cap A_i$. Bidder $i$ could have taken the set $A'_i$ at price $\Sigma_{j\in A'_i}p_j$. Observing that $v_i(A'_i)\geq \Sigma_{j\in A'_i}q_j \geq 2\Sigma_{j\in A'_i}p_j$ we get that the profit of $i$ from taking $A'_i$ is at least $\Sigma_{j\in A'_i}q_j/2$. Since the bundle $T_i$ was a most profitable bundle for $i$, its profit must be at least $\Sigma_{j\in A'_i}q_j/2$ and in particular we have that $v_i(T_i)\geq \Sigma_{j\in A'_i}q_j/2$. Summing up over all bidders in $N'$ we get that $\Sigma_{i\in N'}v_i(T_i)\geq \Sigma_{i\in N'}\Sigma_{j\in A'_i}q_j/2=\Sigma_{j\in W}q_j/2$.
\end{proof}

When $E[\Sigma_{j\in W}q_j]> \frac {o\cdot E[c]} {2}$ (where expectation is taken over the random choice of $N'$ and its internal order), the claim implies that $E[\Sigma_iv_i(T_i)]\geq E[\Sigma_{j\in W}q_j/2] > \frac {o\cdot E[c]} 4 \geq \frac {o\cdot E[c]\cdot r} {4}$, as needed (the last inequality holds since $r\leq 1$). 

From now on we assume that $E[\Sigma_{j\in W}q_j]\leq \frac {o\cdot E[c]} {2}$. We prove the claim by running the following thought experiment. Choose uniformly at random a bidder $i'\in N-N'$ to arrive to the fixed price auction just after the auction with the bidders in $N'$ has ended. We will see that the expected profit of $i'$ is is at least $\frac {o\cdot E[c]} {2n}$, where expectation is taken over the choices of $N'$, its internal order, and $i'$. The expected profit of bidders (in $N'$) that arrive before $i'$ is at least the expected profit of $i'$, simply because the set of available items does not increase during the auction. Since $T_i$ is a profit maximizing bundle of each bidder $i$ and since $v_i(T_i)$ is at least its profit, we get that $E[v_i(T_i)]\geq \frac {o\cdot E[c]} {2n}$. By linearity of expectation we conclude that $E[\Sigma_{i}v_i(T_i)]\geq  n\cdot r\cdot \frac {o\cdot E[c]} {4n}\geq r\cdot \frac {o\cdot E[c]} {4}  $, as claimed. Thus, to finish the proof it remains to bound the profit of $i'$:

\begin{claim}
The expected profit of bidder $i'$ is at least $\frac {o\cdot E[c]} {2n}$, where expectation is taken over the choices of $N'$, the internal order of $N'$, and $i'$.
\end{claim}
\begin{proof}
When $i'$ arrives the set of available items is $M-\cup_iT_i$. The expected value of items that are allocated in $A$ to bidders that are not in $N'$ but were not allocated in the fixed price auction is:
\begin{align*}
E[\Sigma_{j\in (\cup_{i\notin N'}A_i-\cup_{i\in N'}T_i)}q_j]&=E[\Sigma_{j\in \cup_{i\in N}A_i}q_j]-E[\Sigma_i\Sigma_{j\in A_i\cap(\cup_kT_k)}q_j]-E[\Sigma_{j\in W}q_j] \\
&\geq o-E[\Sigma_i\Sigma_{j\in A_i\cap(\cup_kT_k)}q_j]-\frac {o\cdot E[c]} {2}\\
&= o\cdot E[c]-\frac {o\cdot E[c]} {2}=\\
&\frac {o\cdot E[c]} {2}
\end{align*}
Let $O_{i'}=A_{i'}-\cup_kT_k$. Since bidder $i'$ is selected uniformly at random from $N-N'$, we have that $E[v_{i'}(O_{i'})]\geq \frac {o\cdot E[c]} {2n}$.

Similarly to the proof of Claim \ref{claim-contribution-of-W}, observe that for every $j$, $q_j\geq 2p_j$, and therefore when $i'$ arrives to the auction his profit from taking the bundle $O_{i'}$ is at least $\Sigma_{j\in O_{i'}}\frac {q_j} 2 \leq \frac {v_{i'}(O_{i'})} 2$. Since $T_{i'}$ is one of $i'$'s most profitable bundles, $T_{i'}$ must be at least as profitable as $O_{i'}$. %Therefore, $v_{i'}(T_{i'})$ must be at least the profit from $O_{i'}$: $v_{i'}(T_{i'})\geq \frac {v_{i'}(O_{i'})} 2$. Thus we have that $E[v_{i'}(T_{i'})]\geq \frac {o\cdot E[c]} {4n}$.
\end{proof}

This finishes the proof of Lemma \ref{lemma:random-set-posted-price}.
%The expected contribution of bidders that arrive before $i'$ is at most the expected contribution of $i'$, simply because the set of items that were not taken does not increase during the auction. By linearity of expectation we get that $E[\Sigma_{i}v_i(T_i)]\geq  n\cdot r\cdot \frac {o\cdot E[c]} {4n}\geq r\cdot \frac {o\cdot E[c]} {4}  $, as claimed.
\end{proof}

\section{Analysis of the Approximation Ratio: The Main Proof}

We put each instance in one of two sets, and show how to get the desired approximation ratio in each set. Let $(O_1,\ldots, O_n)$ be an optimal allocation and denote its value by $OPT=\Sigma_iv_i(O_i)$. Bidder $i$ is \emph{dominant} if $v_i(M)\geq \frac {OPT} {\sqrt { \log m}}$.

\subsection{Case I: There is a Dominant Bidder}

The first set of instances that we consider is the set of instances with at least one dominant bidder. It is easy to get a good approximation in these instances: with some constant probability in Step \ref{step:second-price} we run a second price auction and therefore the bundle of all items will be allocated to bidder that maximizes $v_i(M)$. Bidder $i$ is obviously a dominant bidder. This provides the promised approximation ratio.

\subsection{Case II: No Dominant Bidder}

Suppose that there is no dominant bidder. We assume that the mechanism does not terminate with the second-price auction of Step \ref{step:second-price} and our analysis is conditioned on this event, which occurs with constant probability.

Fix some choice of bidders in STAT (this determines also the allocation APX). Let $O^{FU}=(O^{FU}_1,\ldots, O^{FU}_n)$ be some allocation restricted to bidders in $FINAL\cup UNIFORM$. We also assume that $O^{FU}$ is restricted to bins with prices identical to the prices associated with bins $1,\ldots,4\log m$ of the allocation APX. For every bidder $i$ and item $j\in O^{FU}_i$, let $q_j$ denote the supporting price of $j$ in $O^{FU}_i$. We sometimes abuse notation a bit and use $O^{FU}=\Sigma_iv_i(O^{FU}_i)$.

\begin{definition}
The \emph{value} of chunk $C_k$ (denoted $val(C_k)$) is the sum of the contributions of the bins in $C_k$ in the allocation $(O^{FU}_1, \ldots, O^{FU}_n)$.
\end{definition}

\begin{definition}
We say that chunk $C_k$ was \emph{reached} if there was an iteration of the loop of Step \ref{step:loop} of the mechanism where $C_k$ was considered (i.e., the mechanism did not terminate earlier).
\end{definition}

\begin{claim}\label{claim-probability-reached}
All chunks are reached with probability at least $\frac 1 {e}$.
\end{claim}
\begin{proof}
If chunk $C_k$ was not reached then the mechanism terminated at Step \ref{step:count-and-posted-price} in one of the previous iterations. The probability that the mechanism terminates in a given iteration is $\frac 1 \alpha$ (independently of all other events). Thus the probability that the auction did not terminate in all previous iterations is $(1-\frac 1 \alpha)^{k}\geq (1-\frac 1 \alpha)^{\alpha}\approx \frac 1 e$. In particular, the chunk with the largest index is reached with probability at least $\frac 1 {e}$. However, for this chunk to be reached all chunks with smaller indices must be reached as well. 
\end{proof}

\begin{definition}
Consider a chunk $C_k$. We say that chunk $C_k$ is \emph{easily approximable} if $E[\Sigma_iv_i(T^k_i)]\geq \frac {val(C_k)} {32}$, where $T^k$ is the allocation constructed in Step \ref{step:count-and-posted-price} when $C_k$ is reached and expectation is taken over the random choices of the bidders in UNIFORM and their internal order (after fixing the bidders in STAT).
\end{definition}
Let $NE$ be the set of chunks that are not easily approximable. The analysis of case II is divided into two. In the first part, assume that $\Sigma_{C_k\notin NE}val(C_k)\geq \frac 1 5 \cdot \Sigma_{C_k\in NE}val(C_k)$. We will show that the expected welfare in this case is $\Omega(\frac {O^{FU}} \alpha)$. When $\Sigma_{C_k\notin NE}val(C_k)<\frac 1 5\cdot  \Sigma_{C_k\in NE}val(C_k)$, we will show that the expected welfare of the fixed price auction of Step \ref{step:final-posted-price-auction} (which is reached with probability $\frac 1 e$) is $ \frac {\alpha\cdot O^{FU}} {{64 \log m}}-\alpha \cdot \frac {OPT} {4\log ^2m}$. These statements will hold for any allocation $O^{FU}$ as defined above. We will also see that when choosing the bidders in STAT, with high probability there exists an allocation $O^{FU}$ with $O^{FU}=\Omega(OPT)$. Thus the overall approximation ratio is $O(\alpha)=O(\sqrt {\log m})$.

\subsubsection{Case IIa: The Value of Easiliy Approximable Chunks is Large}

%Let $C_k$ be some chunk. We claim that with constant probability we will reach this chunk: 

Assume that $\Sigma_{C_k\notin NE}val(C_k)\geq \frac 1 5 \cdot \Sigma_{C_k\in NE}val(C_k)$. For every chunk $C_k\notin NE$ define $E_k$ to be the event in which chunk $C_k$ is reached and the mechanism terminates then in Step \ref{step:count-and-posted-price}. Let $W_k=\Pr[E_k]\cdot E[\Sigma_iv_i(T^k_i)]$.

Observe that since the $E_k$'s are disjoint the expected welfare of the mechanism is at least $\Sigma_{k\notin NE}W_k$. We now compute a lower bound to $W_k$. By Claim \ref{claim-probability-reached}, $C_k$ is reached with probability $\frac 1 {e}$. If $C_k$ is reached, then with probability $\frac 1 \alpha$ (independently of any other events) the auction ends with the allocation $(T^k_1,\ldots,T^k_n)$. Since $C_k$ is easily approximable in this case we have that in this case $E[\Sigma_iv_i(T^k_i)]\geq  \frac {val(C_k)} {32}$. Thus, for every chunk $C_k\notin NE$, $W_k= \frac {val(C_k)} {32e\cdot \alpha}$. Considering that we assume that $\Sigma_{C_k\notin NE}val(C_k)\geq \frac {O^{FU}} 6$, we have that the expected welfare of the mechanism is at least $\Sigma_{C_k\notin NE}W_k= \Sigma_{C_k\notin NE}\frac {val(C_k)} {32e\cdot\alpha}= \frac {O^{FU}} {192\cdot e\cdot \alpha}$.

\begin{example}
Consider a combinatorial auction with $n$ additive bidders and $n$ items. The value of bidder $i$ for each item $j\neq i$ is $v_i(\{j\})=0$. For $i=j$ we have that $v_i(\{j\})>0$. For concreteness, let us suppose that each $v_i(\{j\})$ is chosen uniformly at random from the set $\{1,2,4,8,\ldots, m\}$. Clearly, the optimal solution is to give item $i$ to each player $i$. Denote the value of the optimal allocation by $OPT$. Observe that with high probability the value of the optimal welfare restricted to UNIFORM is $\Theta(OPT)$. Moreover every chunk is easily approximable since item $i$ is demanded only by player $i$. The approximation ratio is therefore $O(\sqrt {\log m})$.

One could attempt to complicate the example and add bidders with ``low'' values to prevent bidders to take the item they demand by increasing the competition in the fixed price (say, adding a bidder with valuation $v(S)=|S|$). This might make some chunks not easily approximable. The analysis of the next case shows that we can take advantage of intensive competition by using the bidders in UNIFORM to infer some information about the prices of items in the optimal solution restricted to bidders in FINAL.
\end{example}

\subsubsection{Case IIb: The Value of Easily Approximable Chunks is Small}

We now assume that $\Sigma_{C_k\notin NE}val(C_k)< \frac 1 5 \cdot\Sigma_{C_k\in NE}val(C_k)$. In the analysis of this case we assume that all chunks were reached. By Claim \ref{claim-probability-reached} this happens with probability $\frac 1 e$. We need more notation. 

%\begin{definition}
%A chunk $C_k$ is \emph{good} if it is not easily approximable and if the value it supports is at least $\frac {O^{FU}} 6$.
%\end{definition}

%\begin{claim}
%Let $NE'$ be the set of good chunks. $\Sigma_{C_k\in NE'}val(C_k)\geq \frac {2\cdot O^{FU}} {3}$.
%\end{claim}
%\begin{proof}
%Let $sup(C_k)$ denote the value supported by chunk $C_k$ in $O^{FU}$ and observe that for every $k$ $sup(C_k)=sup(C_{k+1})+val(C_k)$. Thus, if we let $NE''$ be the set of chunks with supported value at least $\frac {O^{FU}} 6$ we have that $\Sigma_{C_k\in NE''}val(C_k)\geq \frac {5O^{FU}} 6$. Recall that by our assumption $\Sigma_{C_k\in NE}val(C_k)\geq \frac {5\cdot O^{FU}} 6$ and observe that $NE'=NE\cap NE''$. We get that $\Sigma_{C_k\in NE'}val(C_k)\geq \frac {2\cdot O^{FU}} {3}$.
%\end{proof}

\begin{definition}
Consider some chunk $C_k$. Let $M_k$ be the set of items in the allocation $O^{FU}$ that are in bins that belong to $C_k$. Let $(T^k_1,\ldots, T^k_n)$ be the outcome of the fixed price auction of Step \ref{step:count-and-posted-price} when $C_k$ is reached. The set of \emph{available items} of chunk $C_k$ is $L_k=M_k\cap (\cup_iT^k_i)$. The \emph{available value} of chunk $C_k$ is $\Sigma_{j\in L_k}q_j$. 
\end{definition}

\begin{claim}\label{claim-available-value}
The expected available value of every chunk $C_k$ that is not approximable is at least $\frac {3val(C_k)} 4$. Moreover, the expected sum of available value in all chunks is at least $\frac {5O^{FU}} {8}$. 
\end{claim}
\begin{proof}
We use Lemma \ref{lemma:random-set-posted-price} for the proof of this claim. Using the notation of the statement of Lemma \ref{lemma:random-set-posted-price}, we let $(A_1,\ldots, A_n)$ be the allocation $O^{FU}$ restricted to bins in $C_k$. $C_k$ is not approximable hence $E[\Sigma_iv_i(T^k_i)]< \frac {val(C_k)} {32}$. Lemma \ref{lemma:random-set-posted-price} gives that (again using the notation of the statement of the lemma and observing that $r=\frac 1 2$, since UNIFORM is a random group that consists of half of the bidders in $FINAL\cup UNIFORM$):
$$
\frac {val(C_k)} {32}\geq \frac {val(C_k)\cdot E[c]} 8 \Longrightarrow E[c]\leq \frac 1 4 
$$
Recall that $c=1-\frac {\Sigma_i\Sigma_{j\in A_i\cap(\cup_kT_k)}q_j} {\Sigma_iv_i(A_i)}$, and thus we get that $E[\Sigma_i\Sigma_{j\in A_i\cap(\cup_kT_k)}q_j]\geq \frac 3 4\cdot \Sigma_iv_i(A_i)$. The first part of the claim follows since $M_k=\cup_iA_i$. The second part of the claim follows by using linearity of expectation and our assumption that $\Sigma_{C_k\in NE}val(C_k)\geq \frac {5O^{FU}} 6 $.
%Using $\Sigma_{j\in \cup_i A_i}q_j\geq \frac {O^{FU}} 6$ and since  $E[c]=1-\frac {E[\Sigma_{j\in \cup_iA_i-\cup_iT^k_i}q_j]} {\Sigma_{j\in \cup_i A_i}q_j} $ we get that $E[\Sigma_{j\in \cup_iT^k_i}q_j]\geq \Sigma_{j\in \cup_iA_i}q_j-\frac {val(C_k)} 3$. Since $val(C_k)=\Sigma_{j\in M_k}q_j$ we have that $E[\Sigma_{j\in M_k\cap(\cup_iT^k_i)}q_j]\geq \frac {2val(C_k)} 3$.
\end{proof}

A bin that belongs to chunk $C_k$ is \emph{colored} if it was selected in Step \ref{step:pricing}. The set of \emph{correctly priced items} of a bin that was colored is the set of items $COL_k$ in that bin that are also in $L_k$ (the set of available items of $C_k$). Correctly priced items are essentially items that we ``guessed'' their supporting price in the allocation $O^{FU}$ correctly. The next claim shows that we can correctly guess a large chunk of any set $R$ that consists of items that are all available. Eventually, we will set $R$ to be the set of items in $O^{FU}$ that belong to bidders in FINAL.

\begin{claim}\label{claim-value-correctly-priced}
Let $R\subseteq \cup_kL_k$ be some set of available items. $E[\Sigma_{j\in R\cap(\cup_kCOL_k)}q_j]\geq \frac {\alpha\cdot \Sigma_{j\in R}q_j} {{4\cdot \log m}}$.
\end{claim}
\begin{proof}
%Consider a chunk $C_k$ that is not easily approximable. By Claim \ref{claim-available-value} its expected available value is at least $\frac {3val(C_k)} 4$. 
Since there are $\frac {4\log m} {\alpha}$ bins in each chunk and the colored bin is chosen uniformly at random, the expected value of the correctly priced items of chunk $C_k$ that are also in $R$ is $\frac {\alpha\cdot \Sigma_{j\in R\cap L_k}q_j} {4\cdot \log m}$. The claim now follows by using linearity of expectation.
\end{proof}

It is possible that we correctly guessed the price of a certain item, but unfortunately re-priced it in one of the next iterations. The next lemma bounds the value of items ``lost'' due to this incorrect re-pricing.

\begin{lemma}\label{lemma:repricing}
Consider an iteration $k$ of the mechanism and let $M'_k$ be the set of items that belong to some chunk $C_{k'}$, $k'<k$ and were repriced at price $p'$ in Step \ref{step:pricing} of iteration $k$. If $\alpha< \frac {\log m} {\log\log m}$ then $\Sigma_{j\in M'_k}q_j\leq \frac {OPT} {2\log ^2m}$.
\end{lemma}
\begin{proof}
Observe that since in each chunk only the $r$'th smallest bin is colored, there are least $\frac {4\log m} \alpha$ bins that separate the bin of every item $j\in M'_k$ and the currently colored bin. The price doubles in each consecutive bin, and therefore $p'\geq p\cdot 2^{\frac {4\log m} \alpha}\geq p\cdot \log^2m$, where $p$ is the largest price of an item in $M'_k$ before $C_k$ is considered. Since the welfare of any allocation is at most $OPT$: 
$$
OPT\geq \Sigma_iv_i(T^k_i)\geq  |M'_k|\cdot p'\geq |M'_k|\cdot p \cdot \log^2m
$$
Where in the second inequality we use $|M'_k|\subseteq \cup_iT^k_i$ (because each item in $M'_k$ was allocated in the fixed price auction of iteration $k$) and the fact that the allocation of the fixed price auction of iteration $k$ is supported by $p'$. Now observe that $2\Sigma_{j\in M'_k}q_j\leq |M'_k|\cdot p$, since items in $M'_k$ were correctly colored with price at most $p$. Together we get that $\Sigma_{j\in M'_k}q_j\leq \frac {OPT} {2\log^2 m}$.
%
%
%
%Denote by $m''$ the number of items in bin $p$. Since bin $p$ is a good guess and its value is at least $\frac {O_{FINAL}} {16\log m}$, we have that: 
%\begin{align}\label{eq:bin-pricing}
%m''\cdot 2p \geq \frac {O_{FINAL}} {16\log m}
%\end{align}
%Let $p'$ be the price that was considered in the next iteration of the loop. 
%
%
%Consider some iteration of the mechanism and let $m'$ the number of items that were repriced then in Step \ref{step:pricing}. Since we consider only odd or only even chunks, we have that $p'\in C_{k+2}$. In particular, there are at least $\frac {2\log m} \alpha$ bins that belong to chunk $C_{k+1}$ and all of them have prices between $p$ and $p'$. The price doubles in each consecutive bin, thus we have that $p'\geq p\cdot 2^{\frac {2\log m} \alpha}\geq p\cdot \log^2m$.
%
%The contribution of every bin to any allocation is at most $OPT$ and therefore $m'\cdot p\cdot \log^2m\leq OPT$. Together with (\ref{eq:bin-pricing}) we get:
%\begin{align*}
%m'\cdot p\cdot \log^2m&\leq  m''\cdot p \cdot 16\log m\\
%\frac {m'} {m''}&\leq \frac {16 } {\log m}
%\end{align*}
%I.e, a fraction of at most $\frac {16} {\log m}$ of the items in bin $p$ will be repriced in the next iteration. Similarly, in the iteration that follows it at most $\frac {16m'} {\log m}$ items will be re-priced. Thus the overall number of items that will be repriced is bounded from above by a geometric series with a subconstant multiplier, which implies that at most half of the items of $m''$ will be repriced, as claimed.
\end{proof}

Now we are ready to bound the welfare of the final fixed price auction.

\begin{lemma}
Suppose that the mechanism reaches Step \ref{step:final-posted-price-auction}. Let $A=(A_1,\ldots,A_n)$ be a random variable that denotes a welfare-maximizing allocation among all allocations that are restricted to bidders in $FINAL$ and are supported by $p_1,\ldots,p_m$. Then, $E[\Sigma_{j\in \cup_i A_i}p_j]\geq \frac {\alpha\cdot O^{FU}} {32\cdot \log m}- \alpha \cdot\frac {OPT} {2\log ^2m}$.  
\end{lemma}
\begin{proof}
In this proof we denote by $(O^{F}_1,\ldots,O^{F}_n)$ the allocation $(O^{FU}_1,\ldots,O^{FU}_n)$ restricted to bidders in FINAL. Observe that since FINAL is a random set of bidders we have that $E[\Sigma_iv_i(O^F_i)]= \frac {O^{FU}} {2}$. Let $R= (\cup_iO^F_i)\cap (\cup_kL_k)$ be the set of items that are allocated in the allocation $O^F$ and also belong to the set of available items. By Claim \ref{claim-available-value}, $E[\Sigma_{j\in \cup_kL_k}q_j]\geq \frac {5O^{FU}} {8}$. We therefore claim that $E[\Sigma_{j\in R}q_j]= E[\Sigma_{j\in \cup_kL_k}q_j]-E[\Sigma_{j\notin \cup_iO_i^F}q_j]\geq \frac {5O^{FU}} {8}-\frac {O^{FU}} {2}=\frac {O^{FU}} {8}$. Denote the set of items in $R$ that are correctly priced by $R'$. By Claim \ref{claim-value-correctly-priced}, $E[\Sigma_{j\in R'}q_j]\geq \frac {\alpha\cdot \frac {O^{FU}} 8} {{4\cdot \log m}}=\frac {\alpha\cdot O^{FU}} {{32\cdot \log m}}$. By Lemma \ref{lemma:repricing} at each one of the $ \alpha $ iterations we lose at most $\frac {OPT} {2\log ^2m}$ of that value.

Let $G$ be the set of items that were correctly priced and not repriced in later iterations. For each $i$, let $A_i=(O^{F}_i\cap T')-G$. Observe that $(A_1,\ldots, A_n)$ is supported by $p_1,\ldots, p_m$. By the discussion above: 
$$
E[\Sigma_{j\in \cup_iA_i}p_j]\geq \frac {\alpha\cdot O^{FU}} {{32\cdot \log m}}-\alpha \cdot \frac {OPT} {2\log ^2m}
$$
\end{proof}

\subsubsection{Concluding Case II: Existence of an Allocation with High Value}

The next lemma is a ``standard'' random sampling argument with a small addition. Roughly speaking, it considers randomly partitioning a set of bidders to two sets of (almost) equal size -- STAT and $UNIFORM \cup FINAL$. The lemma says that when there is no dominant bidder, we will get a good approximation to the optimum in both sets, even when ignoring items with very low contribution (where contribution is according to the value of the items in the XOS clause). The proof is simple: standard random sampling arguments tell us that the optimal solution restricted to STAT and to $FINAL\cup UNIFORM$ have roughly the same value ($OPT^S$ and $OPT^{FU}$). However, we have to ignore the contribution of items that are too small (say, below $OPT^{FU}/m^2$) since the mechanism considers only logarithmically many prices. The point is that $APX/m^2$ is in the same ballpark of $OPT^{FU}/m^2$, so we can use the former as a good approximation to the latter without losing much, recalling that APX is an $O(1)$ approximation to $OPT^S$.

\begin{lemma}\label{lemma-good-statistics}
With probability $1-o(1)$ there is an allocation $(O_1^{FU}, \ldots, O_n^{FU})$ to bidders in $FINAL\cup UNIFORM$ that is restricted to bins with prices identical to the prices associated with bins $1,\ldots ,4\log m$ of the allocation APX such that $\Sigma_iv_i(O^{FU}_i)\geq \frac {OPT} 8$.
\end{lemma}
\begin{proof}
Let $(OPT_1^{FU}, \ldots, OPT_n^{FU})$ be an optimal allocation restricted to bidders in $FINAL\cup UNIFORM$ and let $(OPT_1^{S}, \ldots, OPT_n^{S})$ be the optimal allocation restricted to bidders in $STAT$. Applying Lemma \ref{lemma:random-set-posted-price} twice and using the union bound, with probability $1-o(1)$ it holds that:
$$
\Sigma_iv_i(OPT_{i}^{FU})\geq \frac {OPT} 6, {\Sigma_iv_i(OPT_{i}^S)} \geq \frac {OPT} 6
$$
In particular, the welfare of the allocation $OPT^{FU}$ is high. It remains to show that when restricting $OPT^{FU}$ to bins $1,\ldots, 4\log m$ of $APX$ we do not lose much. Towards this end, recall that $OPT\geq \Sigma_iv_i(APX_i)\geq \frac {\Sigma_iv_i(OPT_{i}^S)} 2\geq \frac {OPT} {12}$. Thus if the price associated with bin $k$ in the allocation APX is $p(k)$, and the price associated with bin $k'$ in the allocation $OPT^S$ is also $p(k)$, then $k+3\geq k'\geq k-3$. Moreover the difference between every two bins in $APX$ and $OPT^S$ that share the same price is the same and equals $k-k'$.

Observe that for any allocation $A$, bins with indices larger $2\log m+3$ are in fact empty in any allocation, since the contribution of every item that belongs to such bin is greater than $|A|$, which is impossible. Since we restrict our attention to bins $1, \ldots, 4\log m$, the only items in $OPT^S$ that we might miss are those that are in bins with indices at most $4$. The next claim shows that this loss is bounded.

\begin{claim}\label{claim-significant-bins}
Fix an allocation $A$ and let $P$ be the set of items in bins with indices at most $4$. $\Sigma_{j\in P}q_j\leq \frac {16|A|} {m}$.
\end{claim}
\begin{proof}
For every item $j$ in bin $k\leq 4$, it holds that $q_j\leq \frac {16|A|} {m^2}$. We have at most $m$ items and so the total value of items in these bins is at most $m\cdot \frac {16|A|} {m^2}=\frac {16|A|} {m}$.
\end{proof}

We can finally define the allocation $O^{FU}$ from the statement of the lemma: it is the allocation $OPT^{FU}$ restricted to bins $5,\ldots, 4\log m$. By Claim \ref{claim-significant-bins} we have that $\Sigma_iv_i(O^{FU}_i)\geq \Sigma_iv_i(OPT^{FU}_i) - \frac {16\Sigma_iv_i(OPT^{FU}_i)} {m}\geq \frac {3\Sigma_iv_i(OPT^{FU}_i)} {4}\geq \frac {OPT} 6 \cdot \frac 3 4 =\frac {OPT} 8$.
\end{proof}

We now conclude the proof of the approximation ratio for case II. We showed that the expected welfare is at least $\frac{O^{FU}} {192e\cdot \alpha}$ or that we have found prices such that there is an allocation to bidders in FINAL such that the supported value by these prices is at least $\frac {\alpha\cdot O^{FU}} {{32\cdot \log m}}- \alpha \cdot \frac {OPT} {2\log ^2m}$. By Lemma \ref{lemma-good-statistics} with very high probability $O^{FU}=\frac {OPT} 8$, thus in the first case we get an expected approximation ratio of $O(\alpha)=O(\sqrt {\log m})$. In the second case we have that the mechanism reaches the final fixed price auction with probability $\frac 1 e$, and in this case Lemma \ref{lemma:xos:revenue} gives us that the expected welfare is at least:
\begin{align*}
\frac 1 2 \cdot \left (\frac {\alpha\cdot O^{FU}} {{32\cdot \log m}}-\alpha \cdot \frac {OPT} {2\log ^2m} \right) &\geq \frac 1 2 \cdot \left (\frac {\sqrt{\log m}\cdot OPT} {{256\cdot \log m}}-{\sqrt {\log m}} \cdot \frac {OPT} {2\log ^2m}\right )\\
&\geq \frac {\sqrt{\log m}\cdot OPT} {{1024\cdot \log m}}=\\
&\frac { OPT} {{1024\cdot \sqrt{\log m}}}
\end{align*}
which gives us $O(\sqrt{\log m})$ approximation in this case as well.

\section{Implementation for Budget Additive Bidders}

We would like now to implement the mechanism for budget additive bidders. The mechanism from the previous sections works for general XOS functions but requires access to value and demand queries. To apply it to budget additive valuations, we need to show how to efficiently simulate these queries when the valuations are explicitly given to us as input. 

In general, computing a demand query for budget additive bidders is NP hard. However, observe that the algorithm uses only demand queries with prices that can be written $p_j=c\cdot t_j$, where $c$ is some constant and $t_j$ is some integer between $1$ and $poly(m)$. We show that these demand queries can be computed by a polynomial time algorithm, which implies a polynomial time implementation of the mechanism for budget additive bidders. For simplicity we assume below that $c=1$. This is without loss of generality since for every budget additive valuation $v$ we can run the algorithm below for the budget additive valuation $v'$ where $v'(s)=v(s)/c$.

\begin{proposition}
Let $v$ be a budget additive valuation. Let $p_1,\ldots, p_m$ be integers between $1$ and $poly(m)$. A bundle $S$ that maximizes $v(S)-\Sigma_{j\in S}p_j$ can be found in time $poly(m)$.
\end{proposition}
\begin{proof}
Notice that the sum of prices of every combination of items can get only $poly(m)$ values. Denote the possible sums by $q_1,\ldots, q_t$ ($t=poly(m)$). We now use a dynamic programming to find the profit maximizing bundle of items. Define a matrix $A$ of dimensions $|M|\times t$, with the intention to write in cell $A(j,k)$ a value maximizing set $S\subseteq \{1,\ldots,j\}$ such that $\Sigma_{j'\in S}p_{j'}\leq q_k$. Cells in the first row can be easily computed, and we fill in the next rows one by one using $A(j,k)=\arg\max (v(A(j,k-1)), v(A(j-1,k')+\{j\}))$, where $k'$ is such that $q_{k'}=q_k-p_j$.
\end{proof}

\begin{corollary}
There is a universally truthful mechanism for budget additive bidders that provides an expected approximation ratio of $O(\sqrt {\log m})$.
\end{corollary}

\bibliographystyle{plain}
\bibliography{bib}

\appendix

\section{Missing Proofs}

\subsubsection*{Proof of Lemma \ref{lemma-statistics}}

Recall the following tail bound:

\begin{proposition}[Hoeffding bound]\label{claim:general:hoeffding}
Let $X_1,\ldots, X_n$ be independent random variables, such that for each $i$ we have that $X_i\in [a_i,b_i]$. Let $\overline X=\frac {\Sigma_iX_i} n$. Then, 
$$
\Pr[|\overline X-E[\overline X]|\geq t]\leq 2e^{-\frac {2n^2t^2} {\Sigma_i(b_i-a_i)^2}}
$$
\end{proposition}

\begin{corollary}\label{claim:general:weighted-chernoff}
Let $X_1,\ldots, X_n$ be independent random variables such that $X_i=b_i$ with probability $p$ and $0$ with probability $1-p$. Suppose that $b_i\leq l$ for all $i$. Let $X=\Sigma_iX_i$. Then, 
$$
\Pr[|X-E[X]|>\alpha\cdot E[X]] \leq 2e^{-\frac {2\alpha^2\cdot p \cdot E[x]} {l}}
$$
\end{corollary}
\begin{proof}(of Corollary \ref{claim:general:weighted-chernoff})
Let $t=\frac {\alpha\cdot E[X]} n$. Applying Proposition \ref{claim:general:hoeffding}:
\begin{align*}
\Pr[|X-E[X]|>\alpha\cdot E[X]] &\leq 2e^{-\frac {2n^2(\frac {\alpha\cdot E[X]} {n})^2} {\Sigma_i(b_i)^2}} \leq
 2e^{-\frac {2(\alpha\cdot E[X])^2} {\frac {E[x]} {l\cdot p}\cdot l^2 }}
 =2e^{-\frac {2\alpha^2\cdot p \cdot E[x]} {l}}
\end{align*}
\end{proof}

For every bidder $k$ denote by $A_k$ the random variable that receives the value $v_k(T_k)$ if $k\in S$ and $0$ otherwise. Let $A=\Sigma_k A_k$. We will show that with the specified probability $A\geq p\cdot \Sigma_k\frac {v_k(T_k)} 2$.

Since every bidder belongs to $S$ with probability $p$ we have that $E[A]= p\cdot \Sigma_kv_k(T_k)$. By the conditions of the lemma for each $k$ we have that $A_k< \frac {\Sigma_kv_k(T_k)} {R}$. Hence, by Corollary \ref{claim:general:weighted-chernoff}:
\begin{align*}
\Pr[A<p\cdot \frac {\Sigma_kv_k(T_k)} 2 ]&\leq \Pr[|A- p\cdot \Sigma_kv_k(T_k)|\geq p\cdot \frac {\Sigma_kv_k(T_k)} 2 ] \leq 2e^{-\frac {p\cdot R} {2}}
\end{align*}

\subsubsection*{Proof of Lemma \ref{lemma:xos:revenue}}

For every bidder $i$, let $W_i=\cup_{i'<i}S_{i'}$ denote the set of items that were allocated before bidder $i$ arrives to the auction.

Let $OPT_i=\Sigma_{j\in (\cup_{i'\geq i}T_i)-W_i}p'_j$. Observe that $OPT_1=\Sigma_{j\in\cup_iT_i}p'_j$ and that $OPT_{n+1}=0$.  We will show that for every $i$, $2\cdot v_i(S_i)\geq OPT_{i}-OPT_{i+1}$, which implies the lemma.

Notice that for every bidder $i$, $W_{i+1}=W_i+S_i$ and that the allocation $A_i=(\emptyset, \ldots, \emptyset ,T_{i}-W_i,T_{i+1}-W_i,\ldots,T_n-W_i)$ is still supported by $p'_1,\ldots,p'_m$. Therefore, $OPT_{i}-OPT_{i+1}=\Sigma_{j\in (T_i-W_i)}p'_j+\Sigma_{j\in S_i}p'_j$. %It remains to bound from above each of the two terms separately. % by $2v_i(S_i)$.

We finish the proof by showing that $OPT_{i}-OPT_{i+1}\leq 2 v_i(S_i)$. To see that, observe that bidder $i$ could gain a profit of at least $v_i(T_i-W_i)-\Sigma_{j\in (T_i-W_i)} p_j=v_i(T_i-W_i)-\Sigma_{j\in (T_i-W_i)} \frac {p'_j} 2\geq \Sigma_{j\in (T_i-W_i)} \frac {p'_j} 2$ by choosing $T_i-W_i$ (using the fact that $p'_1,\ldots,p'_m$ support the allocation $A_i$). Since $S_i$ is one of bidder $i$'s most profitable sets, it is at least as profitable as $T_i-W_i$, i.e., $v_i(S_i)-\Sigma_{j\in S_i}\frac {p'_j} 2 \geq \frac {v_i(T_i-W_i)} 2$. We have that $v_i(S_i)\geq \frac {v_i(T_i-W_i)} 2 +\Sigma_{j\in S_i}\frac {p'_j} 2\geq \Sigma_{j\in (T_i-W_i)} \frac {p'_j} 2+\Sigma_{j\in S_i}\frac {p'_j} 2=\frac {OPT_{i}-OPT_{i+1}} 2$. 

%We now show that $\Sigma_{j\in S_i}p'_j\leq 2v_i(S_i)$. To see this,  observe that bidder $i$'s profit is non-negative, that is $v_i(S_i)- \Sigma_{j\in S_i}p_j\geq 0$. Now we have that $v_i(S_i)\geq \Sigma_{j\in S_i}p_j\geq \Sigma_{j\in S_i}\frac {p'_j} 2$.

%We also need to account for losses due to being unable to fully allocate $T_i$ to each bidder $i$ with $T_i\cap S_1\neq \emptyset$. Define the allocation $T'=(T'_1,\ldots, T'_n)$ where for every bidder $i$ we have that $T'_i=T_i-S_1$. Observe that for every bidder $i$ and set of items $T\subseteq T'_i$ we have that $v_i(T)\geq \Sigma_{j\in T}p_j$, by the definition of a supporting clause of an XOS valuation. Thus the allocation $T'$ is still supported by $p_1,\ldots , p_m$ and moreover we can bound the revenue loss: $\Sigma_i\Sigma_{j\in T_i}p_j-\Sigma_i\Sigma_{j\in T'_i}p_j=\Sigma_{j\in S_1}p_j$.

%To conclude, by assigning $T_1$ to bidder $1$ we lose a revenue of $O(\Sigma_{j\in T_1}p_j)$ while having that $v_1(S_1) =\Omega(\Sigma_{j\in T_1}p_j)$. The analysis continues similarly by considering the allocation $T'$ and the set of items $S_2$ bidder $2$ receives in the fixed price auction, and so on.

\end{document}